\documentclass[preprint,authoryear,12pt]{article}

\usepackage{newlfont}
\usepackage{a4wide}
\usepackage{amsmath}
\usepackage{inputenc}
\usepackage{amsfonts}
\usepackage{amssymb}
\usepackage{amsthm}
\usepackage{newlfont}
\usepackage{graphicx}
\usepackage{subfigure}
\usepackage[authoryear]{natbib}
\usepackage[normalem]{ulem}
\usepackage{color}

\newtheorem{theorem}{Theorem}

\newtheorem{proposition}{Proposition}
\newtheorem{definition}{Definition\rm}
\newtheorem{remark}{Remark}

\theoremstyle{definition}
\newtheorem{example}{Example}

\newcommand{\PP}{{\mathcal P}}

\newcommand{\BB}{{\mathfrak B}}

\newcommand{\B}{{\mathcal B}}

\newcommand{\IR}{{\mathrm{I\!R}}}

\newcommand{\C}{{\mathcal C}}

\newcommand{\IN}{{\mathrm{I\!N}}}

\newcommand{\TT}{{\mathcal T}}

\newcommand{\co}{{/\!/}}
\inputencoding{ansinew}
\begin{document}

\title{On the core of normal form games with a continuum of players~: a correction\footnote{This is an amended version of the original paper with the same title published in : \emph{Mathematical Social Sciences 89(2017), 32-42.}
}}
\author{Y. Askoura}
\maketitle

{{LEMMA, Universit\'e Paris I\!I, Panth\'eon-Assas, 4 rue Blaise Desgoffe, 75006 Paris, France. email~: {youcef.askoura@u-paris2.fr}}}

\begin{abstract}
We study the core of normal form games with a continuum of players and without side payments. We consider the weak-core concept, which is an approximation of the core, introduced by Weber, Shapley and Shubik. For payoffs depending on the players' strategy profile, we prove that the weak-core is nonempty. The existence result establishes a weak-core element as a limit of elements in {weak-cores} of appropriate finite games.
We establish by examples that our regularity hypotheses are relevant in the continuum case and the weak-core can be strictly larger than the Aumann's $\alpha$-core. For games where payoffs depend on the distribution of players' strategy profile, we prove that analogous regularity conditions ensuring the existence of pure strategy Nash equilibria are irrelevant for the non-vacuity of the weak-core. 
\end{abstract}

\textbf{keywords} :
Weak-core; $\alpha-$core; Game with a continuum of players; Large anonymous games; Normal form games.


\section{Introduction}

It is commonly admitted that the continuum property on the set of agents is valuable in economy and clearly inferring results from finite games is not possible or at least not obvious. Hence the study of the continuum case may be an interesting task \emph{per se} and requires, in the absolute, proper techniques. For some type of continuum games, it is possible, using adequate techniques, to prove analogous results to that known in finite-player games. Some other type of games with a continuum of players may involve, for the existence of pure strategy equilibria, unusual regularity conditions. For more accuracy, let us summarize roughly two models of continuum games prevalent in the literature. Denote by $(T,\TT,\mu)$ a probability space of players, and $A$ a common action space for the players. Under some other regularly assumptions making all the involved entities meaningful, a strategy profile is a function $f:T\rightarrow A$. The payoffs can be given by $U(t,f(t),f)$, a game which we denote $U$, or by $W(t,f(t),\mu f^{-1})$, a game we denote $W$. For the game $U$, the existence of a Nash equilibrium \citep{KHA85,KHA86,BAL95,BAL99} is obtained using expected conditions, mainly~: the concavity of $U(t,a,\cdot), (t,a)\in T\times A$, and some appropriate continuity hypothesis of $U(t,\cdot,\cdot),t\in T$. However, in case of the game $W$, \citet{SCH} shows that a pure-strategy Nash equilibrium exists when $A$ is finite and $\mu$ nonatomic. More generally, the countability of the action space $A$ is necessary for the existence of pure-strategy Nash equilibrium \citep{KhS95,KRS,KhS02} for $W$. This is an unexpected condition, or may be viewed as a new condition illustrating that it may happen that we can learn more in continuum games comparatively to finite-player games. 

\medskip
The Nash equilibrium is extensively investigated in the literature of normal form games with a continuum of players. Whereas, this is not the case for the $\alpha-$core. Up to now, many questions remain unsolved for the $\alpha-$core in the continuum situation. Recall that the $\alpha$-core of normal form games is introduced by \citet{AUM61}. Its main existence result, for games with a finite set of players, is established by \citet{SCA71}. Some generalizations to games of different aspects, but always a finite set of players are obtained in \citep{KAJ92,UYA15,ASK15,AST13}. When considering a continuum set of players, this concept must be somewhat approximated. \citet{WEB81} introduced an adequate approximation of the core, called weak-core, in the setting of games with a continuum of players in a characteristic function form. 

\medskip 
In this paper, we focus on the weak-core in the setting of strategic normal form games with a continuum of players and without side payments, as initiated in \citep{ASK11}. We succeeded to prove an existence result for a particular case of payoffs for the game $U$. We assume that the payoffs depend only on the strategy profile of the players, more precisely, payoffs of the form $U(t,f)$. Together with the concavity assumption and some other regularity conditions on payoffs, we establish an existence result. 

\medskip
Notwithstanding their conceptual differences and their targeted economic situations, the Nash equilibrium on one hand and the $\alpha-$core on the other involve, for their existence, analogous regularity conditions\footnote{See Section \ref{DIS}.} for games with a finite set of players.  We prove, in this paper, that this is not the case for games with a continuum set of players of type $W$. We provide an example of a game satisfying all the conditions required for the existence of the pure-strategy Nash equilibrium, but its weak-core is empty. Moreover, this example satisfies additional conditions that may be expected to be necessary for the weak-core. In this example the payoffs are more general, they depend on individual actions and on the distribution of the strategy profile of the players. This proves that for general payoffs, the weak-core may be empty and must require (for its existence) more restrictive hypotheses. 

For payoffs depending on the strategy profile and individual actions of the players (the game $U$), we expect that the weak-core can be empty under reasonable and intuitive conditions, even similar to that we use hereafter to study particular payoffs. For this kind of payoffs we do not succeeded to provide counter-examples. Note that these two case studies (payoffs depending directly on strategies or on their distributions) may be different when convexity assumptions are used.

\medskip
Properly speaking, the core is already approximated in order to relax the convexity assumption in finite large games \citep{ShS66}. It is shown that the convexity (of preference sets) is not very important when the set of agents becomes large enough. In transferable utility case, \citet{ShS66} obtained interesting results (non-vacuity of approximate cores) for exchange economies with sufficiently many participants. This approximation type is successfully applied later for (finite large) games induced from \textit{pergames} (a configuration in which the payoff achievable by a coalition is a function of the number of its players and their characteristics or attributes) in \citep{WOO83,WOO94,WOO08} and \citep{WoZ84}. In the case of a continuum set of players, \citet{WEB79} studied this concept (approximation in a same direction for the core) and proved its existence for balanced games in characteristic function form and without side payments. A slightly different approximation of the core is introduced by \citet{KAN70,KAN72} and applied for finite markets. \citet{HSZ73} carried on this direction in order to generalize the Shapley-Shubik results for large economies without side payments. Other studies on the approximation of the core can be found in (\citet{AND85} and \citet{STAR69}).

\medskip 
Weber's approximation \citep{WEB81} we deal with, in this work, is different from the above.  
Its purpose is not the overcoming of the nonexistence of the exact core resulting from the non convexity assumption, but to comply with the continuum case and upper semi-continuous payoffs. 
For games with a finite set of players {and continuous payoffs}, our weak-core is exactly the $\alpha-$core of Aumann. We provide, in Section \ref{SEX}, an example 
of game satisfying all the used conditions for the non vacuity of the weak-core, in which the $\alpha-$core is empty, thus legitimating the introduction of the approximation. 
Even for the weak-core approximation, a second exemple shows that an additional regularity condition on payoffs with respect to players must be assumed, in order to handle the continuum framework. This role is played by the \emph{equi-upper-semicontinuity} of players' utility functions, a fact that does not show up in games with finite sets of players. In fact, for finite games, this assumption is satisfied automatically, provided all utilities are upper-semicontinuous.

\medskip
For non exhaustive list of works devoted to games with a continuum of players, the reader is referred to \citet{WEB81,WEB79}, \citet{IcW}, \citet{ROS} and \citet{KaW96} for utility characteristic function form games. For more general frameworks concerning (exchange) economies and markets, we can  cite the famous works of \citet{AUM64,AUM66}, \citet{HIL74,HIL68}, \citet{MAS}, \citet{HHK} and \citet{KhY81}. In a non-cooperative setting, the Nash equilibrium is particularly investigated in \citet{SCH}, \citet{MAS84}, \citet{KHA89}, \citet{RAT} and \citet{KRS}.

\section{Preliminaries, overall framework}\label{PRE}

Let $(T,\TT,\mu)$ be a probability space, where $\TT$ is a $\sigma$-algebra on $T$ and $\mu$ a $\sigma-$additive probability measure on $\TT$. The space $T$ refers to the set of players. Denote by $\TT^+$ the set of elements of $\TT$ with strictly positive $\mu$-measure. The sets of $\TT^+$ refer to coalitions. Note that two coalitions with $\mu$-null symmetric difference will be confused. 

Let $A$ be a convex compact subset of a separable\footnote{The separability of $X$ is not relevant because one can consider the subspace spanned by $A$ which is always a separable Banach space.} Banach space $X$. The set $A$ represents the space of actions. It is common to all players. Denote $\BB(X)$ the Borel $\sigma$-algebra of $X$.

Let $\BB(T,X)$ be the set of all $\TT - \BB(X)$ measurable essentially bounded functions from $T$ to $X$. The space of all measurable functions from $T$ to $A$ (that are automatically essentially bounded) is denoted by $\BB(T,A)$. It refers to the space of pure strategy profiles. The terms ``essentially bounded'' refer to the boundedness relatively to the essential supremum norm  $\|\cdot\|_{\BB(T,X)}$ on $\BB(T,X)$ abbreviated esssup-norm and defined by~: $$\|f\|_{\BB(T,X)}= \sup\left\{\alpha\in \IR : \mu\{t\in T,\|f(t)\|>\alpha\} >0\right\}$$
where $\|\cdot\|$ is the norm of $X$.  Since $A$ is compact, $\BB(T,A)$ is a norm closed subset of $\BB(T,X)$. In all this paper the complement of a subset $E\subset T$ into $T$ is denoted $\complement E$. The notation $X\setminus Y$ refers also to the complement of $Y$ to $X$.

Let $E\in \TT$. Denote $\BB(E,X)$ the set of all measurable essentially bounded functions from $E$ to $X$. $\BB(E,A)$ is defined similarly and refers to the set of strategies of the coalition $E$. For $f_E\in \BB(E,X)$ and $f_{\complement E}\in \BB(\complement E,X)$, denote $f_E/\!/f_{\complement E}$ the function defined by $f_E$ on $E$ and $f_{\complement E}$ on $\complement E$. $0_E$ denotes the almost everywhere null function of $\BB(E,X)$. Since $E$ is measurable, for every measurable function $f_E\in \BB(E,X)$, $f_E/\!/0_{\complement E}\in \BB(T,X)$ and for every $f\in \BB(T,X)$, $f_{|E}\in \BB(E,X)$. Hence, by identifying $\BB(E,X)$ with the subspace $\{f_E/\!/0_{\complement E} : f_E\in \BB(E,X)\}$, $\BB(T,X)$ can be represented as the algebraic direct sum~: $$\BB(T,X)=\BB(E,X)\oplus\BB(\complement E,X)$$

Note that this is also true for a finite number of factors. That is, for a given pairwise disjoint finite family $E_i\in\TT$, $i$ in a finite set $I$, such that $\underset{i\in I}{\cup} E_i=T$ up to a $\mu$-null set, we have the direct algebraic sum~:
\begin{equation}\label{ALG_SUM}
\BB(T,X)=\underset{i\in I}{\oplus} \BB(E_i,X)
\end{equation}

Endow $\BB(T,X)$ with a locally convex topology $\tau_{\BB(T,X)}$ satisfying the following statement~:
\begin{itemize}
\item[(C)] for every  $E\in \TT$, $\BB(E,A)$ is compact in $\BB(T,X)$ for $\tau_{\BB(T,X)}$.
\end{itemize}

We will see in the appendix that this condition is satisfied by the usual weak topologies. Denote by $\tau_{\BB(E,X)}$ the induced topology (from $\tau_{\BB(T,X)}$) on $\BB(E,X)$. The condition (C) means that $\BB(E,A)$ is compact for $\tau_{\BB(E,X)}$. In all this paper, if it is not expressly mentioned, all subspaces and subsets are endowed with the induced topology and all products with the  product topology. The abbreviation lsc means lower semi-continuous and usc means upper semi-continuous.

\section{The weak-core for a strategic normal form game}
A (strategic) normal form game with a continuum of players is defined by means of a function $U:T\times A \times\BB(T,A)\rightarrow \IR$.
As interpreted above, $T$ stands for the set of players, $A$ the common set of actions and $\BB(T,A)$ the set of strategy profiles. The payoffs are summarized in $U$ itself. When each player $t$ chooses his strategy $f(t)\in A$, we obtain a function $f:T\rightarrow A$. Assuming that this function is measurable, $f\in \BB(T,A)$, each player $t$ receives the gain $U(t,f(t),f)$.

Denote the obtained game by the triple $$G=((T,\mu),U,A)$$ or simply by $U$.

The weak-core defined in \citep{WEB81} and its blocking concept are adapted to continuum games in normal form of type $U$ as follows~:

\begin{definition}\label{BLS} We say that a coalition $E\in \TT^+$ blocks the strategy $h\in \BB(T,A)$, if there exist $\varepsilon>0$ and a strategy $f_E\in \BB(E,A)$, such that for all $f_{\complement E}\in \BB(\complement E,A)$,
$$U(t,f(t),f_E\co f_{\complement E})>U(t,h(t),h) +\varepsilon  \text{ a.e. on } E$$

The weak-\emph{core} of the game $G$ is the set of strategies that are not blocked by any coalition $E\in \TT^+$.
\end{definition}

Another conceivable model is to consider the action space as a compact metric space $A_0$ and to define the game by means of a function $W:T\times A_0\times \PP(A_0)\rightarrow \IR$. A convexity structure is not needed on $A_0$ (it may be finite following the needs). We denoted by $\PP(A_0)$ the set of probability measures on $A_0$. The space $\PP(A_0)$ is endowed as frequently with the weak (star) topology.

To each player $t$ corresponds a utility function $W(t,\cdot,\cdot):A_0\times\PP(A_0) \rightarrow \IR$. For this game, the set of strategy profiles $\BB(T,A_0)$ stands simply for $\TT-\BB(A_0)$ measurable functions. Under the strategy profile $f\in \BB(T,A_0)$, each player $t$ receives the gain $W(t,f(t),\mu f^{-1})$. Here, $\mu f^{-1}$ stands for the image probability of $\mu$ under the function $f$.

Denote the obtained game by the triple $$H=((T,\mu),W,A_0)$$ or, as previously, simply by $W$.

The previously defined blocking concept can be stated in an analogous way~:

\begin{definition}\label{BLD} For the game $H$, we say that a coalition $E\in \TT^+$ blocks the strategy $h\in \BB(T,A_0)$, if there exist $\varepsilon>0$ and a strategy $f_E\in \BB(E,A_0)$, such that for every $f_{\complement E}\in \BB(\complement E,A_0)$~:
$$W(t,f_E(t),\mu (f_E\co f_{\complement E})^{-1})>W(t,h(t), \mu h^{-1}) +\varepsilon  \text{ a.e. on } E$$

The weak-\emph{core} of the game $H$ is the set of strategies that are not blocked by any coalition $E\in \TT^+$.
\end{definition}

\begin{remark} By removing $\varepsilon$ from Definitions \ref{BLS} and \ref{BLD}, we obtain naturally the adapted Aumann's $\alpha$-core \citep{AUM61} and its corresponding blocking concept. Observe, trivially, in both cases that the weak-core contains the $\alpha-$core.
\end{remark}

Definitions \ref{BLS} and \ref{BLD} assert that a coalition blocks a given strategy of the game if it possesses a strategy making almost all its members better off, at least by some $\varepsilon>0$, regardless of the opponent coalition choices for strategy. Like any core concept, the weak-core describes stable situations in which no coalition has any incentive to form by playing a different strategy. Indeed, it cannot improve upon, relatively to the equilibrium strategy, the payoffs of almost all its members.

\citet{ASK11} studied a reduced form of the game $H$, where the payoffs do not depend on individual actions. Under the strategy profile $f\in \BB(T,A_0)$, it is assigned to each player $t$ the payoff $W(t,\mu f^{-1})$. Then, some topological regularity conditions and a concavity condition on the distribution argument of $W$ ensured an existence result. In Section \ref{DIS}, we discuss the game $H$ with some general payoffs.

In the following section we focus on a reduced form of $G$ with particular payoffs of the form $U(t,f)$ independent on individual actions\footnote{Note that the payoff form $U(t,f)$ does not generalize the form $W(t,\mu f^{-1})$ if a concavity assumption is needed on the second argument of $W$ and $U$. Indeed $f\mapsto W(t,\mu f^{-1})$ need not be concave in $f$ if $W(t, \cdot)$ is concave, even if we endow the action space $A_0$ with a convexity structure which is necessary before asking for the concavity of $W(t,\mu f^{-1})$ in $f$. That is the results proved in \citep{ASK11} and Theorem \ref{THM} of section \ref{STR} are different, because, each one uses a concavity assumption on the second argument of the payoffs. There is other differences related to the nature of the action spaces and the regularity assumptions on payoffs.}. We give further examples legitimating the introduction of the approximation ``weak-core" and discussing the used assumptions. We expect that the weak-core of $G$ may be empty for a general form of payoffs such as $U(t,f(t),f)$, under ``reasonable" regularity conditions. However, we do not succeeded to provide a counter-example for the latter case.

For the Nash equilibrium,  there is an alternative formulation of a continuum game on characteristics achieved by \citet{MAS84} for games with continuous payoffs, mainly generalized for upper semi-continuous payoffs by \citet{KHA89}. Unfortunately, it appears that the $\alpha-$core (weak-core) cannot comply with an analogous formulation in a straightforward manner.

\section{Payoff as a function of pure strategies}\label{STR}
In this section we consider the game $G$, where the payoff of each player is reduced to depend only on the strategy profile $f\in \BB(T,A)$. That is,  every player $t$ receives the payoff $U(t,f)$. Without more specification when speaking about ``blocking'',  we mean the blocking concept of Definition \ref{BLS} with the reduced payoffs $U(t,f)$, for $(t,f)\in T\times \BB(T,A)$.

\subsection{Existence result}

Let us introduce the following definition~: 
\begin{definition} Let $F$ be a topological space. A family of functions : $v_\gamma :F\rightarrow \IR, \gamma\in \Gamma,$ is said to be equi-usc at $x_0\in F$, if for every $\varepsilon>0$, there exists a neighborhood $V_{x_0}$ of $x_0$ in $F$, such that $$ v_\gamma(x)<v_\gamma(x_0)+\varepsilon,\forall x\in V_{x_0},\forall \gamma\in \Gamma.$$
The family $\{v_\gamma\}_{\gamma\in \Gamma}$ is said to be equi-usc (on $F$) if it is equi-usc at every $x\in F$. 
\end{definition}

We use the following conditions~: 
\begin{itemize}
\item[(R1)] for every $f\in \BB(T,A)$, $t\mapsto U(t,f)$ is measurable, and there exists a $\mu-$integrable function $\psi :T\rightarrow \IR_+$, such that $|U(t,f)|\leq \psi(t)$, for every $t\in T$ and every $f\in \BB(T,A)$.

\item [(R2)] for every $t\in T,$ $f\mapsto U(t,f)$ is concave and the set of utilities $\{U(t,\cdot): t\in T\}$ is equi-usc on $\BB(T,A)$.
\end{itemize}

\begin{theorem}\label{THM}

Under (R1) and (R2), the weak-core of $U$ is nonempty. Furthermore, an element of the weak-core is obtained as a limit of elements in { weak-cores} of appropriate finite games generated from $U$.
\end{theorem}

\begin{proof}  Consider the set $\Pi$ of finite collections of coalitions  containing $T$. Each $\pi\in \Pi,\pi=\{E_i,i\in I_\pi\}$, where $I_\pi$ is finite and one of the sets $E_i=T$. Ordered by inclusion, $\Pi$ is a directed set. Let $\pi=\{E_i,i\in I_\pi\}$ be fixed in $\Pi$. Let $K_j,j\in J_\pi,$ be a finite family of pairwise disjoint measurable subsets of $T$ of strictly positive measure, such that every $E_i,i\in I_\pi,$ is an union (up to a $\mu$-null set) of some sets $K_j$. Naturally, we have $\underset{j\in J_\pi}{\bigcup}K_j=T$ up to a $\mu$-null set.

Considering section \ref{PRE} (and the appendix), for each $j\in J_\pi$, let $Y_j=\BB(K_j,A)$ endowed with the induced topology from $\tau_{\BB(T,X)}$. Denote $Y_\pi=\underset{j\in J_\pi}{\prod} Y_j$ endowed with the resulting product topology. Following the properties of $\tau_{\BB(T,X)}$, all the sets $Y_j,j\in J_\pi,$ and $Y_\pi$ are compact. For each $y=(y_{j_1},...,y_{j_{|J_\pi|}})\in Y_\pi$ corresponds a function $h_y :T\rightarrow A$ defined up to a $\mu$-null set by~: $$h_y=y_{j_1}\co y_{j_2}\co...\co y_{j_{|J_\pi|}}.$$

That is,

$$h_y(t)=\left\{\begin{array}{l}
                  y_j(t), \text{ if } t\in K_j, \\
                  \text{an arbitrary }a\in A \text{ if } t\in T\setminus \underset{j\in J_\pi}\bigcup K_j.
                \end{array}\right.
$$
Since the sets $K_j,j\in J_\pi,$ are pairwise disjoint, $h_y$ is well defined and obviously $h_y\in \BB(T,A)$.

For each $j\in J_\pi$, define a function $g_j :Y_\pi\rightarrow \IR$ by~: $$g_j(y)=\int_{K_j} U(t,h_y) \; d\mu(t), \forall y\in Y_\pi.$$

From the measurability and the boundedness assumption (R1), the functions $g_j$ are well defined and bounded. 

Since $\underset{j\in J_\pi}{\bigcup}K_j=T$ up to a $\mu$-null set, the functions $g_j$ do not depend on the values of $h_y$ defined arbitrarily on $T\setminus \underset{j\in J_\pi}\bigcup K_j$.

Associate with $\pi$ the finite normal form game~:

$$G_\pi=(J_\pi,Y_\pi,\{g_j : j\in J_\pi\})$$

where $J_\pi$ is the set of players, $Y_\pi$ is the product of strategy spaces and $g_j$ is the payoff of the player $j\in J_\pi$.

Note that the functions $g_j$ are  concave on $Y_\pi$. This follows from the  concavity of the functions $U(t,\cdot)$, assumed in (R2) and the linearity of the canonical function $y\mapsto h_y$.
Moreover, the functions $g_j$ are upper semi-continuous. Indeed, let $\{y_\omega\}_{\omega\in \Omega}$ be a net in $Y_\pi$ converging to $y\in Y_\pi$. Let us identify, by using the canonical map described above, the elements of $Y_\pi$ with that of $\BB(T,A)$ and show the upper semi-continuity of the functions $g_j$ on $\BB(T,A)$ (see the property (P) in the appendix).

Let $\varepsilon>0$ be fixed. Using the equi-usc condition (assumed in (R2)), consider a neighborhood $V_\varepsilon(y)$ of $y$ such that,
$$ U(t,y)+\varepsilon>U(t,y'), \forall t\in T,\forall y'\in V_\varepsilon(y).$$

Then there exists $\omega_0 \in \Omega$ such that,
$$ U(t,y)+\varepsilon>U(t,y_w), \forall t\in T,\forall \omega>\omega_0.$$
It results that,

$$\int_{K_j}U(t,y_\omega)\;d\mu< \int_{K_j}U(t,y)\;d\mu+\varepsilon \mu(K_j), \forall j\in J_\pi,\forall \omega>\omega_0.$$

It follows, $$\underset{\omega}{\lim \sup} \int_{K_j}U(t,y_\omega)\;d\mu\leq \int_{K_j}U(t,y)\;d\mu+\varepsilon \mu(K_j), \forall j\in J_\pi.$$

Since $\varepsilon$ is fixed arbitrarily, we conclude that $\underset{\omega}{\lim \sup} \int_{K_j}U(t,y_\omega)\;d\mu\leq \int_{K_j}U(t,y)\;d\mu$, for all $j\in J_\pi$. Hence, the functions $g_j$ are upper semi-continuous.

By the non-vacuity of the $\alpha$-core theorem for games with a finite set of players \citep{SCA71}, and by observing that Scarf's non-vacuity theorem { establishes the existence of elements in the weak-core for games with convex and compact strategy spaces which are subsets of Hausdorff topological vector spaces and usc bounded concave payoffs} (see the appendix), $G_\pi$ has a nonempty { weak-core}. Let $y_\pi$ in the { weak-core} of $G_\pi$ and denote simply $h_\pi=h_{y_\pi}$ the corresponding function in $\BB(T,A)$.

Since $\BB(T,A)$ is compact for $\tau_{\BB(T,X)}$,  the net $h_\pi,\pi\in \Pi,$ has a convergent sub-net, denoted again $h_\pi, \pi\in \Pi$. Denote $h$ the limit of this sub-net.

Let us prove that $h$ belongs to the weak-core of $U$. Assume the opposite. Then, there exists a coalition $E\in \BB_+(T)$ that blocks $h$. From the equi-usc condition, $E$ blocks all strategies in a neighborhood $\theta(h)$ of $h$ by a same strategy. Indeed, there exist $f_E\in \BB(E,A)$ and $\varepsilon>0$, such that for all $f_{\complement E}\in \BB(\complement E,A)$,
$$U(t,f_{E}\co f_{\complement E})> U(t,h) + \varepsilon \text{ for a.e. } t\in E.$$
Using the equi-usc condition, there exists a neighborhood $\theta(h)$ of $h$ such that 
$$U(t,h)+\frac{\varepsilon}{2}>U(t,h'),\forall h'\in \theta (h),\forall t\in T. $$
Then, 
$$U(t,h)+\varepsilon>U(t,h')+\frac{\varepsilon}{2},\forall h'\in \theta (h),\forall t\in T. $$
It results that for all $f_{\complement E}\in \BB(\complement E,A)$, for all $ h'\in \theta (h)$, 
$$U(t,f_{E}\co f_{\complement E})> U(t,h') + \frac{\varepsilon}{2}, \text{ for a.e. } t\in E.$$

Then, there exists $\pi_0\in \Pi$, such that for all $f_{\complement E}\in \BB(\complement E,A)$,

$$U(t,f_{E}\co f_{\complement E})> U(t,h_\pi) +\frac{\varepsilon}{2}, \text{ for a.e. } t\in E, \forall \pi\geq \pi_0.$$

Let $\pi_1=\{E,T\}$ and consider $\pi\in \Pi$ such that $\pi>\pi_1$ and $\pi>\pi_0$. Then, $\pi_1\subset \pi$. In the game $G_\pi$, there is a subset of indices $J_\pi(E)\subset J_\pi$ such that $E=\underset{j\in J_\pi(E)}{\cup}K_j$ up to a $\mu$-null set, where the sets $K_j,j\in J_\pi$, are obtained as above relatively to the present game $G_\pi$. Put for every $j\in J_\pi(E)$, $x_j={f_E}_{|K_j}$. Denote $x_{L}=(x_j)_{j\in L}$, for every $L\subset J_\pi$.

Hence, for every $x_{J_\pi\setminus J_\pi(E)} \in Y_{J_\pi\setminus J_\pi(E)}$, the restriction of $h_{(x_{J_\pi(E)},x_{J_\pi\setminus J_\pi(E)})}$ to $E$ is equal to $f_E$ up to a $\mu$-null set. It results that, for all $x_{J_\pi\setminus J_\pi(E)}\in Y_{J_\pi\setminus J_\pi(E)}$ and all $j\in J_\pi(E)$,

$$  \begin{array}{c} 
g_j(x_{J_\pi(E)},x_{J_\pi\setminus J_\pi(E)})=\int_{K_j} U(t,f_E\co h_{x_{J_\pi\setminus J_\pi(E)}}) \; d\mu(t)>\int_{K_j}U(t,h_\pi)\; d\mu(t){+\delta}=g_j(y_\pi){+\delta}.\\
\end{array}
$$

{Where $\delta\in]0;\; \frac{\varepsilon}{2}\underset{j\in J_\pi(E)}{\min} \mu(K_j)]$, and } $y_\pi \in Y_\pi$ corresponds to $h_\pi$ as explained above and $h_{x_{J_\pi\setminus J_\pi(E)}}$ is the function whose restriction to every $K_j, j\in J_\pi\setminus J_\pi(E),$ equals $x_j$. This means that the coalition $J_\pi(E)$ blocks $y_\pi$ for the {weak-core} blocking concept. This is a contradiction since $y_\pi$ belongs to the {weak-core} of $G_\pi$. This ends the proof.
\end{proof}

\begin{proposition}\label{CONTP} The condition of equi-upper-semicontinuity of $U$ on strategies and the measurability of $U(\cdot,f)$, for every $f\in \BB(T,A)$,  
are satisfied if $T$ is a compact topological space endowed with its Borel $\sigma-$algebra, and $U$ is jointly usc and continuous in $t$.  
\end{proposition}

\begin{proof}The measurability of $U(\cdot,f)$, for every $f\in \BB(T,A)$, results obviously from the continuity of $U$ in $t$. 
Let us prove the equi-usc property. 
Let $\varepsilon>0$ and $f\in \BB(T,A)$. Using the continuity of $U$ with respect to $t$, we can find for every $t\in T$ a neighborhood $V^1_f(t)$ such that,

$$U(t',f)+\varepsilon >U(t,f)+\frac{\varepsilon}{2},\forall t'\in V^1_f(t).$$

From the upper semi-continuity of $U$ on its domain, we can find for every $t\in T$ a neighborhood $V^2_f(t)$ of $t$ and a neighborhood $V_t(f)$ of $f$ such that,

$$U(t,f)+\frac{\varepsilon}{2}>U(t',f'),\forall (t',f')\in V^2_f(t)\times V_t(f).$$

Taking $V_f(t)=V^1_f(t)\cap V^2_f(t)$, we obtain,

$$U(t',f)+\varepsilon>U(t,f)+\frac{\varepsilon}{2} >U(t',f'), \forall (t',f')\in V_f(t)\times V_t(f).$$

When $t$ ranges $T$ we obtain a cover $V_f(t), t\in T$, of $T$ with the corresponding neighborhoods $V_t(f)$ of $f$. Since $T$ is compact, we can extract a finite sub-cover $V_f(t_i)$, $i$ in a finite set $I$.

Put $V(f)=\underset{i\in I}{\cap} V_{t_i}(f)$. Then, for every $t'\in T,$ there exists an index $i$ such that $t'\in V_f(t_i)$. Then, for every $f'\in V(f)$,

$$U(t',f)+\varepsilon>U(t_i,f)+\frac{\varepsilon}{2}>U(t',f'), \forall f'\in V(f).$$
We have constructed a neighborhood $V(f)$ of $f$ such that,
\begin{equation}\label{Lem_Voisin}
U(t,f)+\varepsilon>U(t,f'),\forall f'\in V(f),\forall t\in T.
\end{equation}

\end{proof}
\subsection{The equi-usc hypothesis and the necessity of the approximation of the $\alpha$-core}\label{SEX}
In this section, we provide two examples. The first one shows that Theorem \ref{THM} may fail if the equi-usc condition of $U$ on strategies is relaxed. The second, shows that there may be situations in which all our regularity conditions are satisfied and the $\alpha$-core is empty. Then, this example legitimizes the approximation of the $\alpha$-core. Before stating the examples, let us begin with the common used constructions.

In this section, fix $T=[0,1]$ and $A=[0,1]$.  Consider the Borel $\sigma$-algebra on $T$ and set the previous probability $\mu$ to be the Lebesgue measure $\lambda$ which is a probability on $[0,1]$. We deal here with the particular case of $X=\IR$.  Take $\BB(T,A)\subset L_\infty(T)$ endowed with the weak$^*$ topology. Section \ref{EXTOP}, of the appendix, ensures all topological needs on the spaces. Precisely, the Bochner integral reduces to the Lebesgue integral and Condition (C) is satisfied. Note that $\BB(T,A)$ is metrizable for the weak$^*$ topology, because $L_1(T)$ is separable.

Let $\Gamma:]0,1]\times L_\infty(T)\rightarrow \IR$, defined for every $(t,f)\in ]0,1]\times L_\infty(T)$ by~:
$$ \displaystyle\Gamma(t,f)=\frac{\int_0^t f d\lambda}{t}.$$

Then, for every $t\in ]0,1]$, the function $f\mapsto \Gamma(t,f)$ is linear and continuous on $L_\infty(T)$. It can be represented by an element of $L_1(T)$ in the corresponding duality. That is, $\displaystyle\Gamma(t,\cdot)=\frac{\chi_{[0,t]}}{t}\in L_1(T)\subset L_\infty^*(T)$. Moreover, $\Gamma$ is continuous on $]0,1]\times \BB(T,A)$, as it is obviously sequentially continuous on the metrizable space $]0,1]\times \BB(T,A)$.

Define a second function $G$ on $]0,1]\times \BB(T,A)$, by~:

 $$G(t,f)=\underset{s\in ]0,t]}{\sup}\Gamma(s,f) \text{ for every }(t,f)\in ]0,1]\times \BB(T,A).$$

Then,
\begin{center} $G$ is jointly lsc, convex in $f$ and continuous in $t$. \end{center}

Indeed, it is clear that $G$ is  convex in $f$ as a supremum of linear functions.

Let us prove that it is jointly lsc. Fix $t_0\in ]0,1]$, $f_0\in \BB(T,A)$ and $\alpha\in \IR$ such that $G(t_0,f_0)>\alpha$. Let $\alpha_1\in \IR$ such that $G(t_0,f_0)>\alpha_1>\alpha$. From  the definition of $G$ and the continuity of $\Gamma$ on $]0,1]\times \BB(T,A)$, there is $t_1\in ]0,t_0[$ such that $G(t_0,f_0)<\Gamma(t_1,f_0)+\alpha_1-\alpha$. Using the continuity of $\Gamma$ again, we can find, on $\BB(T,A)$, a neighborhood $V(f_0)$ of $f_0$, such that for all $f'\in V(f_0)$,
$$\Gamma(t_1,f')>G(t_0,f_0)-(\alpha_1-\alpha)>\alpha_1-(\alpha_1-\alpha)=\alpha.$$

Then, for every $f'\in V(f_0)$, for every $t\in ]t_1,1]$, $\underset{s\in]0,t]}{\sup}\Gamma(s,f')>\alpha$.
That is, there exists a neighborhood $V(t_0)=]t_1,1]$ of $t_0$  and a neighborhood $V(f_0)$ of $f_0$ having just defined, such that,
for every $(t,f)\in V(t_0)\times V(f_0)$,
$$G(t,f)>\alpha.$$
This means that $G$ is jointly lsc. It remains to prove that $G$ is continuous with respect to $t$. Let $f\in \BB(T,A)$ be fixed. Since $G$ is jointly lsc, it suffices to prove that $G(\cdot,f)$ is usc. Let $t_0\in]0,1]$ and $\alpha>0$ such that $G(t_0,f)<\alpha$. If $t_0=1$, there is nothing to do, because this implies $G(t',f)<\alpha$, for all $t'\in]0,1]$. Else, let $\alpha_1>0$ such that  $G(t_0,f)=\underset{s\in]0,t_0]}{\sup}\Gamma(s,f)<\alpha_1<\alpha$. Since $\Gamma$ is continuous in $t$, there exists a neighborhood $V(t_0)$ of $t_0$ such that, $\Gamma(t',f)<\alpha_1$, for all $t'\in V(t_0)$. Then, $\underset{s\in V(t_0)}{\sup}\Gamma(s,f)\leq\alpha_1$. Hence,

$$\underset{s\in ]0,t_0]\cup V(t_0)}{\sup}\Gamma(s,f)\leq\alpha_1<\alpha.$$
Choosing $V(t_0)$ to be an open interval centered at $t_0$, we obtain, for every $t'\in V(t_0)$,
$$G(t',f)<\alpha.$$
Then $G(\cdot,f)$ is usc.

\bigskip
The following example shows how the equi-usc of  $U$ on strategies is crucial in Theorem \ref{THM}.
\begin{example}\label{ex-NWC}

The payoff function $U$ is defined as follows~:
$$U(t,f)=\left| \begin{array}{l}
                  \min\left\{\Gamma(t,f),(1-G(t,f))  \displaystyle\frac{1-t}{t} \right\}\qquad \text{ if }t>0, \\
                  1 \qquad \qquad\; \text{ for }t=0.
                 \end{array}
         \right.
 $$

The function $U$ has the following properties~:
\begin{itemize}
\item[(*)] $U$ is bounded, jointly usc on $[0,1]\times \BB(T,A)$ and  concave with respect to its variable $f$.
\end{itemize}
Before proving (*),  note that $U$ satisfies the conditions of Theorem \ref{THM} except the equi-usc condition of $\{U(t,\cdot),t\in T\}$ at $f\equiv 0$.  In fact, since $U(\cdot,f)$ is usc,  it is measurable, the boundedness property completes (R1). The concavity assumption stated in (R2) is expressly mentioned in (*).  

Let us prove (*). It is clear that $U(T\times\BB(T,A))\subset [0,1]$. The function $G$ is lsc on $]0,1]\times \BB(T,A)$. Then, $(t,f)\mapsto 1-G(t,f)$ is usc on $]0,1]\times \BB(T,A)$. Since $t\mapsto \displaystyle\frac {1-t}{t}$ is positive and continuous on $]0,1]$, we can easily verify that $\displaystyle (1-G(t,f))\frac{1-t}{t}$ is usc on $]0,1]\times \BB(T,A)$. Since $G$ is convex in $f$ for every fixed $t\in ]0,1]$, the function $\displaystyle (1-G(t,f))\frac{1-t}{t}$ is concave in $f$ for every fixed $t\in ]0,1]$. Whereas, $\Gamma(t,f)$ is linear in $f$ and continuous on $]0,1]\times \BB(T,A)$. It results that $U$, as a minimum among these two functions, is usc on $]0,1]\times \BB(T,A)$ and  concave in $f$ for every $t\in ]0,1]$. For $t=0$, $U$ is constant in $f$, then  concave as well. Observe now that for every $t\in ]0,1]$ and every $f\in \BB(T,A)$, $U(t,f)\in [0,1]$. Then if $(t_n,f_n)$ is a sequence in $]0,1]\times \BB(T,A)$ converging to $(0,f)$, necessarily $\underset{n}{\lim\sup}U(t_n,f_n)\leq1=U(0,f)$.
Hence, $U$ is also usc at every point of the form $(0,f),$ $f\in \BB(T,A)$.
At this step we proved that $U$ satisfies all the properties listed in (*).

Let us prove now that the weak-core of $U$ is empty.

Let $f$ such that $\underset{t\rightarrow 0}{\lim \sup } \Gamma(t,f)=1$. Let us show that such a strategy cannot be in the weak-core. Indeed, with the assumption $\underset{t\rightarrow 0}{\lim \sup } \Gamma(t,f)=1$, for every $t\in ]0,1]$, $G(t,f)=1$. It results that $U(t,f)=0$ for every $t\in ]0,1]$. If the coalition $E=]0,1/2]$ plays $\displaystyle h_E\equiv \frac{1}{2}$, we obtain, $\displaystyle \Gamma(t,h_E\co h_{\complement E})=G(t,h_E\co h_{\complement E})=\frac{1}{2}$ and $\displaystyle [1-G(t,h_E\co h_{\complement E})]\frac{1-t}{t}=\frac{1}{2}.\frac{1-t}{t}\geq \frac{1}{2}$, for every $h_{\complement E}\in \BB(\complement E,A)$ and every $t\in ]0,1/2]$. Then, $E$ blocks $f$.

Now, let $f\in \BB(T,A)$ such that $\underset{t\rightarrow 0}{\lim \sup } \Gamma(t,f)<1$. Then, there exist $t_0\in ]0,1]$ and $\alpha \in ]0,1[$, such that $\Gamma(t,f)<\alpha$ for every $t\in ]0,t_0]$. Then, for every $t\in ]0,t_0]$, $U(t,f)<\alpha$.

Let $\displaystyle t_1\in ]0,t_0]$ such that $1-t_1>\alpha$. Consider the coalition $E=]0,t_1]$ with its strategy $h_E\equiv 1-t_1$. Then, for every $t\in ]0,t_1]$ and every $h_{\complement E}\in \BB(\complement E,A)$,

$$
\begin{array}{ll}
   & \Gamma(t,h_E\co h_{\complement E})=G(t,h_E\co h_{\complement E})=1-t_1 \text{ and }\\
   & \displaystyle(1-G(t,h_E\co h_{\complement E}))\frac{1-t}{t}=t_1\frac{1-t}{t}=\frac{t_1}{t}.(1-t)\geq 1-t\geq 1-t_1
\end{array}
$$

That is, for every $t\in ]0,t_1]$ and every $h_{\complement E}\in \BB(\complement E,A)$, $U(t, h_E\co h_{\complement E})=1-t_1> \alpha> U(t,f)$. This can be rewritten as~: for every $t\in ]0,t_1]$ and every $h_{\complement E}\in \BB(\complement E,A)$,
$$U(t, h_E\co h_{\complement E})> U(t,f)+(1-t_1)-\alpha.$$
Which means that $E$ blocks $f$. From the foregoing, we can state that the weak-core of $U$ is empty.

\end{example}

The following example establishes that the weak-core may be strictly larger than the $\alpha$-core.

\begin{example}
Here, $U$ is defined as~:
$$U(t,f)=\left| \begin{array}{l}
                  \min\left\{\displaystyle\int_0^t f\;d\lambda,(1-G(t,f))  \displaystyle\frac{1-t}{t} \right\}\qquad \text{ if }t>0, \\
                  0 \qquad \qquad\; \text{ for }t=0.
                 \end{array}
         \right.
 $$

$U$ satisfies all the properties~:
\begin{itemize}
\item[(**)] $U$ is bounded, jointly usc on $[0,1]\times \BB(T,A)$,  concave with respect to its variable $f$, for every fixed $t\in [0,1]$, and continuous in $t$, for every fixed $f\in \BB(T,A)$.
\end{itemize}

Before proving (**), observe, using proposition \ref{CONTP}, that it implies all the conditions of theorem \ref{THM}.

It is clear that $U(T\times\BB(T,A))\subset [0,1]$. Analogously to the previous example, we obtain easily that $U$ is jointly usc on $]0,1]\times \BB(T,A)$ and $U(t,\cdot)$ is concave on $\BB(T,A)$ for every fixed $t\in [0,1]$. Since for every fixed $f\in \BB(T,A)$, $G(\cdot,f)$ is continuous on $]0,1]$, it is clear that $U(\cdot,f)$ is continuous on $]0,1]$ for every fixed $f\in \BB(T,A)$. To achieve the verification of (**), it suffices to prove the continuity of $U$ at $(0,f)$, for every $f\in \BB(T,A)$. In fact, if $(t_n,f_n)$ is a sequence in $]0,1]\times \BB(T,A)$ converging to $(0,f)$, necessarily $\underset{n}{\lim\sup} U(t_n,f_n)\leq \underset{n}{\lim\sup} \int_0^{t_n} f_n \;d\lambda= 0=U(0,f)$. Since $U(t,h)\geq 0$, for every $(t,h)\in [0,1]\times \BB(T,A)$, necessarily $\underset{n}{\lim} U(t_n,f_n)=0=U(0,f)$. Thus, $U$ is continuous at $(0,f)$.

\vspace{12pt}
Let us prove now that the $\alpha$-core of $U$ is empty.
As in the previous example, an element $f\in \BB(T,A)$ such that $\underset{t\rightarrow 0}{\lim \sup } \Gamma(t,f)=1$ cannot be in the $\alpha$-core. For such element $f$, we have necessarily $U(t,f)=0$ for every $t\in ]0,1]$. If the coalition $E=]0,1/2]$ plays $\displaystyle h_E\equiv \frac{1}{2}$, we obtain, for every $h_{\complement E}\in \BB(\complement E,A)$ and every $t\in ]0,1/2]$,
$$\int_0^t h_E\co h_{\complement E}\;d\lambda=\frac{t}{2}\text{ and }[1-G(t,h_E\co h_{\complement E})]\frac{1-t}{t}=\frac{1}{2}.\frac{1-t}{t}\geq \frac{1}{2}$$

Then, $E$ blocks $f$ for the $\alpha$-core blocking concept.

Let $f\in \BB(T,A)$ such $\underset{t\rightarrow 0}{\lim \sup } \Gamma(t,f)<1$.
Take $t_0\in ]0,1/2]$ and $\alpha \in ]0,1[$, such that $\Gamma(t,f)<\alpha$ for every $t\in ]0,t_0]$.
Then, for every $t\in ]0,t_0]$, $G(t,f)\leq \alpha$. That is,
\begin{equation}\label{IN_EXO2_0}
\displaystyle\forall t\in ]0,t_0], (1-G(t,f))\frac{1-t}{t}\geq 1-\alpha
\end{equation}
Since for all $t\in]0,t_0]$, $\displaystyle \Gamma(t,f)=\frac{\int_0^t f\;d\lambda}{t}<\alpha$,

\begin{equation}\label{IN_EXO2_0_b}
\forall t\in]0,t_0], \int_0^t f\;d\lambda<\alpha t
\end{equation}

Let $t_1\in]0,t_0]$ and $\alpha_1>0$ such that,

\begin{equation}\label{IN_EXO2_1}
\alpha t<\alpha_1<1-\alpha, \forall t\in ]0,t_1]
\end{equation}

Then, taking into account \eqref{IN_EXO2_0} and \eqref{IN_EXO2_0_b},
\begin{equation}\label{IN_EXO2_2}
\forall t\in ]0,t_1], U(t,f)=\int_0^t f\;d\lambda<\alpha t
\end{equation}
From this step, it is easy to provide, as in the previous example, coalitions blocking $f$ by constant strategies. We will see thereafter, that such coalitions may possess other type of blocking strategies.

Since $\displaystyle \frac{\int_0^t f d\lambda}{t}<1$, for all $t\in ]0,t_0]$, necessarily $\displaystyle \frac{\int_0^t f-1 \; d\lambda}{t}<0$, for all $t\in ]0,t_0]$. Hence,
\begin{equation}\label{IN_EXO2_3}
\int_0^t 1-f\; d\lambda>0,\forall t\in]0,t_0]
\end{equation}

Let $\delta >0$ to be fixed later.
We have for every $t\in ]0,1]$, $\Gamma(t,f+\delta(1-f))=\Gamma(t,f)+\Gamma(t,\delta(1-f))$. It results, for every $t\in ]0,1]$, $$G(t,f+\delta(1-f))\leq G(t,f)+G(t,\delta(1-f)).$$
Furthermore, for every $t\in ]0,1]$, $$G(t,\delta(1-f))=\underset{s\in ]0,t]}{\sup} \frac{\int_0^s \delta(1-f)\; d\lambda}{s}\leq \underset{s\in ]0,t]}{\sup} \frac{\int_0^s \delta\; d\lambda}{s}=\delta.$$

Gathering the two previous equations, we see that for every $t\in ]0,t_0]$,
$$G(t,f+\delta(1-f))\leq G(t,f)+\delta\leq \alpha+\delta.$$

Taking $\delta_1\in ]0,1]$ such that $1-\alpha-\delta_1>0$,  we obtain for every $\delta\in ]0,\delta_1]$,
\begin{equation}\label{IN_EXO2_4}
\forall t\in ]0,t_0], [1-G(t,f+\delta(1-f))]\frac{1-t}{t}\geq [1-\alpha-\delta]\frac{1-t}{t}\geq 1-\alpha-\delta.
\end{equation}
Consider the coalition $E=]0,t_2]$, $t_2\in ]0,t_1]$ to be fixed later, and its strategies of the form $h_E(\delta)=f_{|E}+\delta (1-f_{|E})$. Remark first that such strategies are feasible for every $\delta\in [0,1]$.
Observe now  that $G(t,h_E(\delta)\co h_{\complement E})= G(t,f+\delta(1-f))$ for every $\delta>0$, $t\in E$ and $h_{\complement E}\in \BB(\complement E,A)$. Then, from \eqref{IN_EXO2_4}, for every $\delta \in ]0,\delta_1]$, $t\in E$ and $h_{\complement E}\in \BB(\complement E,A)$,
\begin{equation}\label{IN_EXO2_5}
[1-G(t,h_E(\delta)\co h_{\complement E})]\frac{1-t}{t}\geq 1-\alpha-\delta.
\end{equation}
In another hand, from \eqref{IN_EXO2_0_b}, for every $\delta>0$ and $t\in E$,
\begin{equation}\label{IN_EXO2_6}
\int_0^t h_E(\delta)\co h_{\complement E}\;d\lambda=\int_0^t f\;d\lambda+\delta\int_0^t(1-f)\;d\lambda\leq  \alpha t+\delta t.
\end{equation}

Choose according to \eqref{IN_EXO2_1}, $\delta\in ]0,\delta_1]$ and $t_2\in ]0,t_1]$ such that,

$$\alpha t+\delta t< \alpha_1< 1-\alpha-\delta,\forall t\in ]0,t_2].$$

This equation together with \eqref{IN_EXO2_5}, \eqref{IN_EXO2_6} and \eqref{IN_EXO2_2} provides, for all $t\in E$, and all $h_{\complement E}\in \BB(\complement E,A)$,

$$U(t,h_E(\delta)\co h_{\complement E})=\int_0^t h_E(\delta)\co h_{\complement E}\:d\lambda=U(t,f)+\delta\int_0^t(1-f)\;d\lambda.$$

That is, using \eqref{IN_EXO2_3}, $U(t,h_E(\delta)\co h_{\complement E})>U(t,f)$, for all $t\in E$ and all $h_{\complement E}\in \BB(\complement E,A)$. Which means that $E$ blocks $f$ with respect to the $\alpha$-core blocking concept. We proved accordingly that the $\alpha$-core of $U$ is empty. However, all the assumptions of Theorem \ref{THM} are satisfied. Then, we can assert that the weak-core of $U$ is nonempty.

\end{example}

\section{Payoff as a function of distributions of strategies : counter-example to a general form of payoffs}\label{DIS}

The question of non-vacuity of the weak-core ($\alpha$-core) under conditions similar to that ensuring the Nash equilibrium arises naturally. In fact, for normal form games with a finite set of players, the regularity conditions guaranteeing the existence of Nash equilibrium and the $\alpha-$core are of ``similar" type\footnote{More precisely, for the $\alpha-$core, we require the convexity and the compactness of strategy spaces, the continuity and the quasi-concavity of payoffs. For Nash equilibrium, we juste weaken the quasi-concavity assumption, assuming it for each player payoff only on its own strategy space. For the $\alpha$-core, the quasi-concavity of payoffs is assumed on the product of all players' strategy sets.}. We emphasize that we do not mean, in any way, any conceptual comparison between the Nash equilibrium and the $\alpha$-core (weak-core). These concepts may be viewed as ``antagonistic" from some economic and game theoretic view point. But, two mathematical problems (existence of these concepts) with similar solutions in a given situation, give raise naturally to the question of their technical comparison in a more general or different situation.

If we adopt the general payoff form $W$ above, it is shown that a pure strategy Nash equilibrium for $H$ exists iff together with some measurability condition, the action space $A_0$ is countable \citep{KhS95,KhS02,KRS}, see Theorem \ref{NAS_EQ} below. Example \ref{EX1}, below, proves that we cannot provide analogous results, even with more restrictive regularity conditions, for the weak-core, then for the $\alpha$-core too, because the weak-core contains the $\alpha$-core. In other words, this section provides some negative answer to the raised question.

When restricting $W$ to depend only on $T\times \PP(A_0)$, an existence result of the weak-core is proved in \citep{ASK11}. The example, thereafter, proves further that this result cannot be generalized directly by adding the player's action argument to the payoff functions.

Before stating this example, let us recall that a pure strategy Nash equilibrium for the game $H$ is a strategy profile $f\in \BB(T,A_0)$ satisfying~:

$$ W(t,f(t),\mu f^{-1})=\underset{a\in A_0}{\max}\; W(t,a,\mu f^{-1})\text{ a.e. on } T.$$

In the sequel $\C (A_0\times \PP(A_0))$ refers to the set of continuous real functions on $A_0\times \PP(A_0)$ endowed with its sup-norm topology and the corresponding Borel $\sigma$-field. Recall that $A_0$ is a compact metric space and $\PP(A_0)$ is the set of Borel probabilities on $A_0$ endowed with its weak$^*$ topology. 

Consider the following version of a well known existence result of pure strategy Nash equilibrium~:

\begin{theorem}\label{NAS_EQ} Assume that $W(t,\cdot,\cdot)$ is continuous for every $t\in T$. Then, the pure strategy Nash equilibrium for the game $H$ exists under the following conditions~:
\begin{itemize}
\item[(a)] the map $t\mapsto W(t,\cdot, \cdot)$ as a function from $T$ into $\C (A_0\times \PP(A_0))$ is measurable,
\item[(b)] $A_0$ is countable,
\item[(c)] $\mu$ is an atomless probability measure.
\end{itemize}
\end{theorem}

This theorem was obtained by \citet{KhS95}. It generalizes the seminal result of \citet{SCH}, where among other $A_0$ is assumed to be finite. 
If $A_0$ is not countable, the pure strategy Nash equilibrium may fail to exist \citep{KhS02,KRS}. Note that we presented here a simplified version of the original result Theorem 10, page 650 in \citep{KhS95}, where among other, the distribution argument of the payoffs is more general.

In order to link our regularity conditions assumed on payoffs (Theorem \ref{THM} above and the main result in \citet{ASK11})  to the measurability condition (a) of Theorem \ref{NAS_EQ}, let us remark that since $A_0\times \PP(A_0)$ is a compact metric space, (a) is satisfied if $W$ is a Carath\'eodory function. That is $W(t,\cdot,\cdot)$ is continuous on $A_0\times \PP(A_0)$, for every $t\in T$, and $W(\cdot,a,p)$ is measurable for every $(a,p)\in A_0\times \PP(A_0)$ (\citet{AlB}, Theorem 4.55, p. 155). If $T$ is a topological space endowed with its Borel $\sigma-$algebra and $\mu$ is a Borel probability, then,
\begin{itemize}
\item Condition (a) of Theorem \ref{NAS_EQ} is satisfied if $W$ is usc in $t$ and jointly continuous in its second and third argument.
\end{itemize}

\begin{example}\label{EX1} In this example, by the word ``blocking'' we mean the blocking concept of Definition \ref{BLD}.
Let $T=[0,1]$ endowed with its Borel $\sigma-$algebra and the Lebesgue probability measure $\lambda$ and $A_0=\{0,1,2,3\}\subset \IR$. $\PP(A_0)$ represents the set of probabilities on $A_0$. Since $A_0$ is finite, the variation norm topology on $\PP(A_0)$ coincides with the used weak (star) topology on $\PP(A_0)$.

Let $\varepsilon>0$ small enough and $E_1,E_2,E_3$ three pairwise disjoint coalitions of $[0,1]$, such that $$\lambda(E_1)=\lambda(E_2)=\lambda(E_3)=\frac{1-2\varepsilon}{3}.$$

The sets $E_i$ are represented in Figure \ref{FIG_EXA} (in every sub-figure).

Let~:

$$p_0\in \PP(A_0) \text{ satisfying } p_0(\{0\})=1. \text{ Put }P_0=\{p_0\}.$$

Put $\displaystyle\alpha=\lambda(E_i)=\frac{1-2\varepsilon}{3}$. For all $i\in \{1,2,3\}$, define the set~:

$$P_i=\{p\in \PP(A_0) : p(\{i\})\geq 2\alpha\}.$$

\begin{figure}[htbp]
\includegraphics[width=0.8\textwidth]{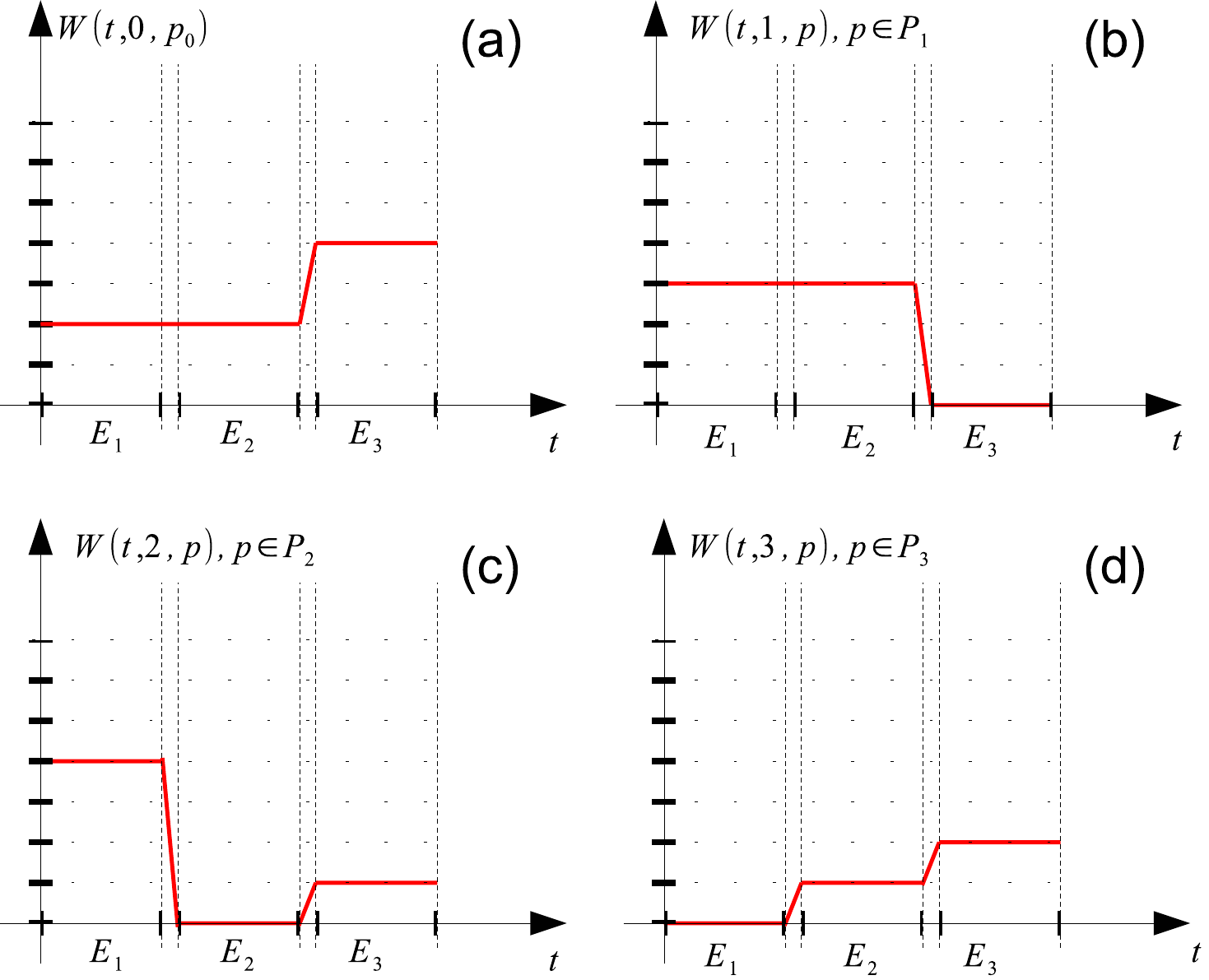}\\
 \caption{Some values of $W$}\label{FIG_EXA}
\end{figure}

Hence, the sets $P_i, i\in \{0,...,3\}$, are closed, convex and pairwise disjoint for $\displaystyle 2\alpha=2\frac{1-2\varepsilon}{3}>\frac{1}{2}$. Then, fix an arbitrary $\varepsilon$ satisfying the previous requirement.

Define $W: T\times A_0\times \PP(A_0)$ as depicted in figure \ref{FIG_EXA} on $\overset{3}{\underset{i=0}{\bigcup}} [0,1]\times\{i\}\times P_i$ and $W(t,a,p)\equiv 0$ elsewhere.

Our final goal is not the game $W$, but in order to facilitate the forthcoming analysis, let us show that its weak-core is empty. First, we can remark that $W$ possesses nice properties. In fact, it is easy to verify that :
\begin{itemize}
\item $W$ is usc on its domain and $W(\cdot,a,p)$ is continuous for every $(a,p)\in A_0\times \PP(A_0)$.
\end{itemize}

\vspace{24pt}
Observe now the following values of $W$~:

$$
\begin{array}{l|l}
  W(t,0,p_0)=\left\{ \begin{array}{l}
                                        2 \text{ if } t\in E_1\cup E_2,\\
                                        4 \text{ for } t\in E_3.
                                      \end{array}
\right.
 & \forall p\in P_1,  W(t,1,p)=\left\{ \begin{array}{l}
                                        3 \text{ if } t\in E_1\cup E_2,\\
                                        0 \text{ if} t\in E_3.
                                      \end{array}
\right. \\
\hline
 \forall p\in P_2,  W(t,2,p)=\left\{ \begin{array}{l}
                                        4 \text{ if } t\in E_1,\\
                                        0 \text{ if } t\in E_2,\\
                                        1 \text{ if } t\in E_3
                                      \end{array}
\right. &  \forall p\in P_3, W(t,3,p)=\left\{ \begin{array}{l}
                                        0 \text{ if } t\in E_1,\\
                                        1 \text{ if } t\in E_2,\\
                                        2 \text{ if } t\in E_3.
                                      \end{array}
\right.
\end{array}
$$

Let any strategy $h\in \BB(T,A_0)$. If $\lambda h^{-1}\notin P_0\cup P_1\cup P_2\cup P_3$, then $W(t,h(t),\lambda h^{-1})=0$, for every $t\in T$. Clearly $T$ blocks $h$ by playing $f_T\equiv 0$ on $T$. In fact, $f_T$ provides the payoffs of sub-figure (a). Now assume that $h$ generates the payoffs of sub-figure (a), that is $\lambda h^{-1}=p_0$. Necessarily, $h\equiv 0$ (a.e.) and then $E_1\cup E_2$ blocks $h$ by $f_{1,2}\equiv 1$ on $E_1\cup E_2$. In fact, for all $f_{\complement E_1\cup E_2}\in \BB(\complement E_1\cup E_2,A_0)$, $\lambda (f_{1,2}\co f_{\complement E_1\cup E_2})^{-1} (\{1\})\geq 2\alpha$. Hence, $E_1\cup E_2$ will obtain the payoffs of sub-figure (b). Analogously, we can show easily that, $E_1\cup E_3$ blocks, by its strategy $f_{1,3}\equiv 2$, all strategy generating the payoffs of sub-figure (b). Thereby, $E_1\cup E_3$ can ensure for its members the payoffs represented in sub-figure (c). Now the coalition $E_2\cup E_3$ can ensure for its members the payoffs represented in sub-figure (d) by playing $f_{2,3}\equiv 3$ on $E_2\cup E_3$, doing this all strategy leading to the payoffs of sub-figure (c) is blocked. In turn, the coalition $T$ blocks all strategies ensuring payoffs of sub-figure (d) by its strategy $f_T\equiv 0$ on $T$. $T$ obtains with the last blocking strategy the payoffs represented in sub-figure (a). Henceforth, we conclude that the weak-core (and then the $\alpha-$core) of the game $W$ is empty.

\vspace{24pt}
\textbf{Emptiness of the weak-core under the continuity of payoffs and their concavity with resect to the distribution argument~:}

Let $W$ defined as in figure \ref{FIG_EXA} on $\overset{3}{\underset{i=0}{\bigcup}} [0,1]\times\{i\}\times P_i$ and consider the function $\widetilde{W}$ on $[0,1]\times A_0\times \PP(A_0)$ defined as follows~:

$$ \widetilde{W}(t,a,p)=W(t,a,P_a)\left(\frac{D}{2}-d(p,P_a)\right)$$

where $D$ stands for the minimum among $d(P_i,P_j), i\neq j$. The metric $d$ is that induced by the variation norm. Since the sets $P_i$ are convex compact and pairwise disjoint, $D$ is strictly positive. We do not need the exact value of $D$.

Let $a$ be fixed. Since, by construction, $W(\cdot,a,P_a)$ is continuous, the continuity of $\widetilde{W}$ on $T\times \{a\}\times \PP(A_0)$ results obviously from the extension formula. Then, $\widetilde{W}$ is continuous on its domain.

It is easy to check, using the convexity of the sets $P_a$, $a\in A_0$, that for every $a\in A_0$, the function $p\mapsto d(p,P_a)$ is convex. If follows that $$p\mapsto W(t,a,P_a)\left(\frac{D}{2}-d(p,P_a)\right)$$
is concave for every $t$ and every $a$. Hence $\widetilde{W}(t,a,\cdot)$ is concave for every $t$ and every $a$.

\vspace{12pt}
\textbf{Emptiness of the weak-core of $\widetilde{W}$.}

Let a strategy $h\in \BB(T,A_0)$ be fixed. Put $p=\lambda h^{-1}$. Two cases can occur~:

\vspace{12pt}
1) Firstly, every $a\in A_0$ satisfies $\displaystyle d(p,P_a)\geq \frac{D}{2}$. In this case, for a.e. $t\in T$, $\displaystyle \frac{D}{2}-d(p,P_{h(t)})\leq 0$. Hence, $\displaystyle \widetilde{W}(t,h(t),p)=W(t,h(t),P_{h(t)})(\frac{D}{2}-d(p,P_{h(t)}))\leq 0$. As previously, many coalitions block $h$. For instance, $T$ blocks $h$, by playing $f_0(t)\equiv 0$. In fact, $\lambda f_0^{-1}=p_0$ and for a.e. $t\in T$,  $d(p_0, P_{f_0(t)})=d(p_0,P_0)=0$, hence every $t\in T$ receives $\displaystyle \widetilde{W}(t,0,p_0)=W(t,0,p_0)\frac{D}{2}$.

\vspace{12pt}
2) Secondly, there is $i_0\in A_0$ such that $\displaystyle d(p,P_{i_0})\leq \frac{D}{2}$. In this case, for all $\displaystyle a\in A_0\setminus \{i_0\}$, $\displaystyle d(p,P_a)\geq \frac{D}{2}$. Hence, the payoff of all player playing $h(t)\neq i_0$ is negative or null, and the payoff of a player $t$ playing $i_0$ is~: $$\widetilde{W}(t,i_0,p)=W(t,i_0,P_{i_0})(\frac{D}{2}-d(p,P_{i_0}))\leq W(t,i_0,P_{i_0})\frac{D}{2}.$$

It follows that for all $t\in T$,
$$\widetilde{W}(t,h(t),\mu h^{-1})\leq W(t,i_0,P_{i_0})\frac{D}{2}.$$
But, in the game $W$, as we seen previously, following the value of $i_0$, either the coalition $T$ or an union of two coalitions among $\{E_1,E_2,E_3\}$ blocks any strategy generating the payoffs $W(\cdot,i_0,P_{i_0})$. Such a coalition blocks also $h$ for the continuous game defined by $\widetilde{W}$. For example, if $i_0=1$, then, $E_1\cup E_3$ can ensure for all its member $t$, by playing $f_{1,3}\equiv 2$ on $E_1\cup E_3$, the payoff $\displaystyle W(t,2,P_2)\frac{D}{2}$ and for all $t\in E_1\cup E_3$ and every $f_{\complement E_1\cup E_3}\in \BB({\complement E_1\cup E_3},A_0)$,

$$\begin{array}{rl}
\widetilde{W}(t,f_{1,3}(t),\lambda (f_{1,3}\co f_{\complement E_1\cup E_3})^{-1})&=W(t,2,P_2)\frac{D}{2}\\
&\geq \left(W(t,1,P_1)+1\right)\frac{D}{2}\\
&\geq \widetilde{W}(t,h(t),\mu h^{-1})+\frac{D}{2}.
\end{array}
$$

We deduce that the weak-core of the game defined by $\widetilde{W}$ is empty. Observe however that $\widetilde{W}$ satisfies all conditions of Theorem \ref{NAS_EQ}, then we know that the set of pure  strategy Nash equilibria of $\widetilde{W}$ is non-empty.

\end{example}
\begin{remark}
The payoffs $\widetilde{W}$ constructed in the previous example satisfy nicer properties than that required for Nash equilibrium in Theorem \ref{NAS_EQ} above. In fact, $\widetilde{W}$ is continuous on its domain and concave with respect to its distribution argument. Regarding Section \ref{STR}, the existence result in \citep{ASK11}, and the classical Scarf's existence result for finite games, these properties of $\widetilde{W}$ may be intuitively expected to guarantee the non vacuity of the weak-core. Nevertheless, this example shows that they do not suffice or they are irrelevant in the present case.
\end{remark}

\section*{Acknowledgments} The author is very grateful to an anonymous referee for his relevant comments.

\section{Appendix}

In this section we give some mathematical considerations as a technical extension of section \ref{PRE}.
\subsection{Concatenating coalition strategies}
We used in this paper the following property of $\tau_{\BB(T,X)}$ (see section \ref{PRE}), which is satisfied by all linear topologies~:

\begin{itemize}
\item [(P)] Let $G:\BB(T,X)\rightarrow \IR$ be continuous (resp. upper semi-continuous) for $\tau_{\BB(T,X)}$. Let $E_i\in \B(T)$, $i$ in a finite set $I=\{1,...,n\}$ be a pairwise disjoint family such that $T=\underset{i\in I}{\cup}E_i$ up to a $\mu$-null set. Let $s:\underset {i\in I}{\prod} \BB(E_i,X)\rightarrow \BB(T,X)$ defined by $s(f_1,...,f_n)= f_1\co f_2...\co f_n$. Then $G\circ s$ is continuous (resp. upper semi-continuous).
\end{itemize}
This results from the continuity of the operation ``+'' for vector topologies (see for instance \citep{SCH71}).

Remark, in the last time, that we can replace, in (P), $X$ by $A$.

\subsection{Topologies satisfying Condition (C)}\label{EXTOP}
In this section, most of used notions and results can be found in \citep{DiU77}.

\vspace{6pt}
\textbf{1)} Assume in this example that $X$ is a separable reflexive Banach space. Let $L_\infty(\mu,X)$ (resp. $L_1(\mu,X)$) be the space of $\mu$-measurable (strongly measurable) essentially bounded (resp. Bochner integrable) functions defined from $T$ to $X$. Since $X$ is a separable Banach space, it is hereditarily separable. Then Pettis's measurability theorem states that every scalarly (weakly) measurable function (then every Borel function) is $\mu$-measurable, because the range of every function taking its values in $X$ is separable. Since $\mu-$measurable functions are Borel, we have $\BB(T,X)=L_\infty(\mu,X)$. For these considerations, we do not distinguish, thereafter between these different notions of measurability.

The finiteness of the measure $\mu$  yields $L_\infty(\mu,X)\subset L_1(\mu,X)$. Consider the embedding of $\BB(T,X)$ in $L_1(\mu,X)$.

We show in the sequel that the weak topology $\sigma(L_1(\mu,X),L_\infty(\mu,X^*))$ satisfies the condition (C). We begin by verifying that (C) is satisfied for $T$.

It is clear that $\BB(T,A)$ is a closed subset of $L_1(\mu,X)$. Indeed a norm convergent sequence of $\BB(T,A)$ has an a.e. convergent subsequence, which provides that the limit is necessarily in $\BB(T,A)$. Being a norm closed subset of $L_1(\mu,X)$ and convex, $\BB(T,A)$ is weakly closed in $L_1(\mu,X)$ as well.
Let us verify the weak compactness of $\BB(T,A)$. Since $A$ is bounded, there is a constant $M=\underset{a\in A}{\sup}\|a\|$. Hence, as $\mu$ is finite, $\BB(T,A)$ is norm bounded in $L_1(\mu, X)$, and,

$$\underset{\mu(E)\rightarrow 0}{\lim}\; \underset{f\in \BB(T,A)}{\sup} \int_E \|f\| \;d\mu\leq \underset{\mu(E)\rightarrow 0}{\lim}M\mu(E)=0.$$

Then, the set $\BB(T,A)$ is uniformly integrable. Since $A$ is convex, for every $f\in \BB(T,A)$ and $E\in \B(T)$, $\int_E f\;d\mu\in \mu(E)A$. That is, $\left\{\int_E f\;d\mu\; :\; f\in \BB(T,A)\right\}$ is a subset of $\mu(E)A$. Consequently, it is relatively compact, then relatively weakly compact too. As $X$ is separable and reflexive, so is $X^*$. It follows that both $X$ and $X^*$ have the Radon-Nikodym property. It results from the Dunford weak compactness criterion in $L_1(\mu,X)$  (\citet{DiU77}, p. 101) that $\BB(T,A)$ is relatively weakly compact. Since this set is weakly closed in $L_1(\mu,X)$, it is weakly compact. Instead of using the Dunford criterion, one can use compactness results in \citep{DIE77,ULG91}.

Verify the condition $(C)$ for any Borel set. Let $E\in \B(T)$. The set $\BB(E,A)$ is identified (see section \ref{PRE}) to the subset $\{f_E\co 0_{\complement E} : f_E\in \BB(E,A)\}$ of $L_1(\mu,X)$. As doing it for $T$, we can observe that $\BB(E,A)$ is weakly closed in $L_1(\mu,X)$. Since $\BB(E,A)\subset \BB(T,A)$, it is necessarily weakly compact.

\vspace{12pt}
\textbf{2)} Assume that $X$ is a separable Hilbert space. Among properties of $X$, we have $X^*=X$. We have already seen in example 1 that $\BB(T,X)=L_\infty(\mu,X)$. We assert that the weak$^*$ topology $\sigma(L_\infty(\mu,X),L_1(\mu,X))$ satisfies the condition (C). Begin the verification of (C) on the set $T$. Let us verify, first, the weak$^*$ closure of $\BB(T,A)$ in $L_\infty(\mu,X)$. Let for this aim $g\in L_\infty(\mu,X)\setminus\BB(T,A)$. Then, there exists a measurable subset $E\in \B(T)$ of strictly positive measure, such that $g(t)\notin A,$ for all $t\in E$. For every $a\in X\setminus A,$ let $\varepsilon>0$ such that $\mathbf{\overline{B}}(a,\varepsilon_a)\cap A=\emptyset$. Here, $\mathbf{\overline{B}}(a,r)$ (resp. $\mathbf{B}(a,r)$) stands for the closed (resp. open) ball of $X$ of radius $r\in \IR_+$ centered at $a$. The set of obtained open balls $\mathbf{B}(a,\varepsilon_a),a\in X\setminus A$, constitutes an open cover of $X\setminus A$. Since $X$ is separable metric, it is hereditarily Lindel\"of. Then, we can extract, from the previous cover, a countable subcover $\mathbf{B}(a_i,\varepsilon_i), i\in \IN$, of $X\setminus A$.
Since $E\subset\underset{i\in \IN}{\bigcup}g^{-1}(\mathbf{B}(a_i,\varepsilon_i))$ and $\mu(E)>0$, there is necessarily an index $i_0\in \IN$ such that $\mu(g^{-1}(\mathbf{B}(a_{i_0},\varepsilon_{i_0})))>0$. By the use of the Hahn-Banach separation theorem, let $x^*\in X^*$ separating strictly $\overline{\mathbf{B}}(a_{i_0},\varepsilon_{i_0})$ and $A$. Since $X$ is an Hilbert space, $x^*$ is represented by an element $x\in X$. Denote by $(\cdot,\cdot)$ the scalar product of $X$. Then, there is $\alpha>0$ such that, $$(x,y)>\alpha>(x,a),\forall y\in \mathbf{B}(a_{i_0},\varepsilon_{i_0}),\forall a\in A.$$

Put $E_{i_0}=g^{-1}(\mathbf{B}(a_{i_0},\varepsilon_{i_0}))$ and $\chi_{i_0}$ its characteristic function. Then $\chi_{i_0} x\in L_1(\mu,X)$ and, $$\int_T(\chi_{i_0} x,g)\; d\mu>\mu(E_{i_0})\alpha>\int_T(\chi_{i_0} x,f)\; d\mu,\forall f\in \BB(T,A).$$
Hence the weak$^*$ open set $\{h\in L_\infty(\mu,X) : (\chi_{i_0} x,h)>\mu(E_{i_0})\alpha\}$ contains $g$ and does not intersect $\BB(T,A)$, which means that $\BB(T,A)$ is weak$^*$ closed in $L_\infty(\mu,X)$. Since $\BB(T,A)$ is bounded for the esssup-norm in $L_\infty(\mu,X)$, it results, from the Banach-Alaoglu Theorem, that  $\BB(T,A)$ is weak$^*$ compact in $L_\infty(\mu,X)$.

Verify (C) for an arbitrary Borel sets. Let $E\in \B(T)$ and denote $L_1(\mu_E,X)$ (resp. $L_\infty(\mu_E,X)$) the set all Bochner integrable (resp. $\mu-$measurable essentially bounded) functions defined from $E$ to $X$. $\mu_E$ stands for the induced measure on $E$. By extending, as in Section \ref{PRE}, all the functions in $L_1(\mu_E,X)$ (resp. $L_\infty(\mu_E,X)$) by $0$ on $\complement E$,  $L_1(\mu_E,X)$ (resp. $L_\infty(\mu_E,X)$) can be seen as a subspace of $L_1(\mu,X)$ (resp. $L_\infty(\mu,X)$). By the same notations and a similar reasoning we have $\BB(E,X)=L_\infty(\mu_E,X)$. As above we obtain easily that $\BB(E,A)$ is weak$^*$ compact in $L_\infty(\mu_E,X)$.

Remark in the last time that the function $f_E\mapsto f_E/\!/0_{\complement E}$ is continuous from $L_\infty(\mu_E,X)$ to $L_\infty(\mu,X)$ both of them endowed with its weak$^*$ topology.
It results that $\BB(E,A)$ is compact for the induced topology on $L_\infty(\mu_E,X)$ from $\sigma(L_\infty(\mu,X),L_1(\mu,X))$.

\subsection{The { weak-core} for usc payoffs in infinite dimension.}
Thereafter, we give the modifications to operate in the Scarf non-vacuity result \citep{SCA71} in order to handle infinite dimensional strategy spaces and bounded usc quasi-concave payoffs {to prove the non-vacuity of the weak-core for a game with a finite set of players}: the assumptions listed in the proof of Theorem 1.

Let $N=\{1,...,n\}$. For every $i\in N$, consider 
\begin{itemize}
\item[(C1)] $X_i$ a convex compact subset of a Hausdorff topological vector space. 
\end{itemize}
Put $X=\underset{i\in N}{\prod}X_i$. Let $S\subset N$. $X_S$ stands for the product $\underset{i\in S}{\prod}X_i$, $v_S$ refers to an element of $X_S$ and $-S=N\setminus S$.

For every $i\in N$,
\begin{itemize}
\item [(C2)]$u_i:X\rightarrow \IR$ is upper semi-continuous,  quasi-concave and bounded from below.
\end{itemize}
Consider the game~: $$ G=(N,X,\{u_i,i\in N\})$$

\citet{SCA71} showed that $G$ has a nonempty $\alpha$-core by assuming that the functions $u_i$ are quasi-concave and continuous, the sets $X_i$ are convex compact subsets of finite dimensional Euclidean spaces. However, his proof remains valid under the weaker conditions (C1) and (C2) above {to prove the existence of the weak-core}. Indeed, Scarf proves that the characteristic function form game $G_c$, defined below, has a nonempty core and to each element in the core of $G_c$ corresponds an element in the $\alpha$-core of $G$ { in case of continuous payoffs. For ucs payoffs, we show at the end that the weak-core remains nonempty}.

The associated characteristic function form game is defined by~:  $$G_c=(N,V),$$
where, for every nonempty $S\subset N$,
$$V(S)=\{y\in \IR^N, \exists v_S\in X_S, u_i(v_S,v_{-S})\geq y_i,\forall v_{-S}\in X_{-S},\forall i\in S\}$$

In order to have a nonempty core (elements of $V(N)$ not belonging to the interior of $V(S)$ for any $S\subset N$), the game $G_c$ needs to satisfy~:
\begin{itemize}
\item[(a)] for every $S\subset N$, $V(S)$ is closed and nonempty,
\item[(b)] for every $S\subset N$, if $y\in V(S)$ and $y'\in \IR^N$ satisfies $y'_i\leq y_i$ for every $i\in N$, then $y'\in V(S)$,
\item[(c)] $V(N)$ is bounded from above.
\item[(d)] $G_c$ is balanced.
\end{itemize}
Prove that $G_c$ satisfy (a)-(c) under the assumptions (C1) and (C2).
The condition (b) is obviously satisfied and (c) results from the upper semicontinuity of the functions $u_i,i\in N$, and the compactness of $X$. Let us prove (a). Fix $S\subset N$. The non-emptiness of $V_S$ results from boundedness of $u_i,i\in N$. Remark that $y\in V(S)$ if and only if there exists $v_S\in  X_S$, such that $u_i(v_S,v_{-S})-y_i\geq 0$, for all $i\in S$, for all $v_{-S}\in X_{-S}$. That is, $$y\in V(S)\Leftrightarrow \exists v_S\in X_S,\underset{v_{-S}\in X_{-S}}{\inf} \; \underset{i\in S}{\min}\left\{u_i(v_S,v_{-S})-y_i\right\} \geq 0 $$

Define the function $H:X_S\times\IR^N$ by~: $$H(v_S,y)=\underset{v_{-S}\in X_{-S}}{\inf} \underset{i\in S}{\min}\left\{u_i(v_S,v_{-S})-y_i \right\}$$
Then, $$y\in V(S)\Leftrightarrow \exists v_S\in X_S, H(v_S,y)\geq 0$$

It is clear that $H$ is upper semicontinuous. Let $y^n$, $n\in \IN$, be a sequence of $V(S)$ converging to $y$. Let $v_S^n$, $n\in \IN$, a corresponding sequence in $X_S$ satisfying $H(v_S^n,y^n)\geq 0$, for every $n\in \IN$.

Since $X_S$ is compact, as a net, $\{v_S^n\}_{n\in \IN}$ has a convergent sub-net, denote it $\{v_S^{\eta(w)}\}_{w\in W}$ and let $v_S$ be its limit.
Then, the sub-net $\{y^{\eta(w)}\}_{w\in W}$ of $\{y^n\}_{n\in \IN}$ converges to $y$ and the net $\{(v_S^{\eta(w)},y^{\eta(w)})\}_{w\in W}$ converges to
$(v_S,y)$ in the product $X_S\times \IR^S$. It follows~:

$$H(v_S,y)\geq\underset{w\in W}{\lim \sup}H(v_S^{\eta(w)},y^{\eta(w)})\geq 0$$

Which means $y\in V(S)$ and then, $V(S)$ is closed.

The remaining arguments of Scarf need not be rewritten. They are also true under our assumptions. For instance, the arguments showing the balancedness of the game $G_c$ (condition (d)) do not require topological considerations, they work with only convexity assumptions taken into account by the quasi-concavity of the functions $u_i$. {Now, take $y$ in the core of $G_c$ and let $\bar v \in X$ such that $u_i(\bar v)\geq y_i$, for all $i\in N$. Then, $\bar v$ is an element of the weak-core of $G$. Otherwise, there is $S$ blocking $\bar v$ with some $v_S\in X_S$. Then, there exists $\varepsilon>0$ such that, 
$$u_i(v_s,v_{-S})>u_i(\bar v)+\varepsilon\geq y_i+\varepsilon,\forall v_{-S}\in X_{-S},\forall i\in S.$$
Hence, $y$ belongs to the interior of $V(S).$ Note that the parameter $\varepsilon$ is needed in the previous formula to conclude that $y$ is in the interior of $V(S)$, since the upper semi-continuity cannot guarantee, for instance, that $\inf_{v_{-S}} u_i(v_s,v_{-S})>y_i$ for every $i\in S$, and the large inequality do not support the needed conclusion in order to provide a contradiction.
\\
Observe that for continuous payoffs, the weak-core coincides with the $\alpha-$core.
}

\bibliographystyle{plainnat}
\bibliography{AskBibMSS3}

\end{document}